\documentclass[twocolumn,pre,superscriptaddress]{revtex4}

\usepackage[english]{babel}
\usepackage{amsmath}
\usepackage{amsthm}
\usepackage{enumitem}
\usepackage[usenames,dvipsnames]{xcolor}

\usepackage{amssymb}
\usepackage{amsmath}
\usepackage{amsthm}
\usepackage[usenames,dvipsnames]{xcolor}
\usepackage{bbm}
\usepackage{cases}

\usepackage{amsmath}
\usepackage{amssymb}
\usepackage{times}
\usepackage{graphicx}
\usepackage{subfigure}
\usepackage[colorlinks=true, allcolors=MidnightBlue]{hyperref}

\newcommand{\Def}{\overset{\mathit{def}}{=}}

\newtheorem{lemma}{Lemma}
\newtheorem{remark}{Remark}
\newtheorem{corollary}{Corollary}

\usepackage{listings}

\begin{document}

\title{Dynamics of the Kac Ring Model with switching scatterers}

\author{Leonid A. Bunimovich}
\email{leonid.bunimovich@math.gatech.edu}
\affiliation{
School of Mathematics, Georgia Institute of Technology, Atlanta, GA 30332, USA.}

\author{Emilio N. M. Cirillo}
\email{emilio.cirillo@uniroma1.it}
\affiliation{Dipartimento di Scienze di Base e Applicate per l'Ingegneria, 
             Sapienza Universit\`a di Roma, 
             via A.\ Scarpa 16, I--00161, Roma, Italy.}

\author{Matteo Colangeli}
\email{matteo.colangeli1@univaq.it}
\affiliation{Dipartimento di Ingegneria e Scienze dell'Informazione e Matematica, Universit\`a degli Studi dell’Aquila, Via Vetoio, I--67100 L’Aquila, Italy.}

\author{Lamberto Rondoni}
\email{lamberto.rondoni@polito.it}
\affiliation{Dipartimento di Scienze Matematiche, Politecnico di Torino, Corso Duca degli Abruzzi 24, I--10129, Turin, Italy.}
\affiliation{INFN, Sezione di Torino,Via Pietro Giuria 1,I--10125, Turin, Italy.}

\begin{abstract}

We introduce a generalized version of the Kac ring model in which particles are of two types, black and white. Black particles modify the environment through which all particles move, thereby inducing indirect and potentially long-range interactions among them. Unlike the inert scatterers of Kac's original model, the scatterers in our setting possess internal states that change upon interaction with black particles and can be interpreted as energy levels of the environment. This makes the model self-consistent, as it incorporates a form of particle interactions, mediated by the environment, that drives the system toward some kind of stationary state. 
Although indirect and long-range interactions do not necessarily promote thermodynamic states, interactions are necessary for energy to be shared among the elementary constituents of matter, enabling the establishment of equipartition, which is a prerequisite for defining temperature. Therefore, our model is one step forward in this direction, elucidating the role of interactions and energy exchange.
We prove that any initial state of the system converges to a time periodic state (i.e. a phase space orbit) and describe basins of attraction for some of such asymptotic periodic states.

\end{abstract}

\maketitle



\section{Introduction}
\label{s:intro} 
\par\noindent

The seminal Kac ring model, introduced by Mark Kac in 1956, was designed to illustrate a role of probabilistic concepts in deriving the second law of thermodynamics from classical mechanics~\cite{Kac1956,Kac1959}; see also~\cite{Lebowitz1973,Spohn1978,Thoms86} for additional insights into the model, and~\cite{DeRoeck2003,Chandrasekaran2012,Jebeile2020,Hiura2021,Gill2023} for recent extensions.

The Kac ring is a deterministic cellular automaton governed by a set of dynamical rules that specify the influence of some scatterers on black and white particles moving around the ring in discrete time steps. Despite its artificial nature, and in the spirit of the earlier stochastic ``dog-flea'' model proposed by P. and T. Ehrenfest \cite{Ehren07}, this model offers valuable insight into the mathematical mechanisms underlying the emergence of irreversibility in many-particle systems. As clearly emphasized in \cite{Bricmont1995}, the Kac ring illustrates the paradigmatic features of a macroscopic system in classical mechanics: it is isolated, reversible and it also exhibits a sort of Poincar\'{e} recurrence. However, apart from exceptional cases that become a vanishing fraction in the large system limit, the recurrence times for any other given system diverge with the number of particles. This expresses in idealized terms Boltzmann's response to Zermelo: macroscopic irreversibility for a single system with microscopic reversible dynamics sets in as the number of particles grows, since any anti-H-theorem behavior is rapidly pushed beyond any physically meaningful timescale. 

A key aspect of the model, in particular, lies in the reinterpretation of Boltzmann's statistical assumption of \textit{molecular chaos}, considering the evolution of ensemble averages. 
The model succeeds to merge, within a unifying picture, the time-reversible dynamics of the particles on the ring with the irreversible behavior of a macroscopic order parameter concerning the system as a whole. The point is that in the large system limit, the statistical (under many respects trivial) irreversibility of ensemble averages becomes the {\em typical} behaviour, which is the behaviour of the vast majority of single systems. The exceptional cases, corresponding to specific choices of initial conditions, become a vanishing fraction of the whole. Naturally, for any finite number of particles, the model reveals that Zermelo's objection to Boltzmann's theory, based on Poincar\'{e}'s recurrence theorem, is valid, since the system returns close to its initial state, after sufficiently long times. However, Boltzmann's reply was that these times are way too long to bear any physical relevance for any macroscopic aggregation of particles, assumed and not granted that the same model could still be applied to describe the system of interest. See e.g. Refs. \cite{Lebowitz1973,Goldstein2012,Pitowsky2012,Chibbaro2014} for extended discussions of these issues.

Although the Kac ring model is well suited to analytical treatment and has clarified several foundational aspects of statistical mechanics, analogously to the Lorentz gas \cite{Bricmont1995,Cohen2002}, its original form lacks particle interactions, hence any kind of energy exchange, which constitutes an essential ingredient for the establishment of local thermodynamic equilibrium. 

Inspired by recent works on stochastic cellular automata with long-range particle interactions \cite{CDP16,CDP17,CDP17b}, which revealed a rich phase diagram, including the presence of metastable regions and phase transitions, we propose an extended version of the classical Kac ring model, which we call the Generalized Kac Ring (GKR). The GKR incorporates novel features absent from the classical version, the most prominent being the introduction of mutual interactions between particles and environment (the scatterers).

In the GKR, particles influence the state of the environment over a characteristic time scale determined by a parameter referred to as the \textit{rigidity of the environment} \cite{Bun04}. The dynamics of particles alter the environment, which in turn modifies their behavior, thereby creating indirect interactions among the particles mediated by their shared environment.
These interactions are non-local and may even be long-ranged. As such, they do not necessarily drive relaxation toward a thermodynamic state, as observed in other models where correlations decay too slowly (see, e.g.,~\cite{BonLebCh}). Nevertheless, they introduce an effective contribution to the energy that influences the dynamics, representing a step toward the key ingredients required for defining temperature and, ultimately, for the establishment of local thermodynamic equilibrium~\cite{Spohn}.

Increasing the rigidity reduces the strength of these effective interactions, and in the limiting case of infinite rigidity the environment ceases to evolve, thereby recovering the original non-interacting Kac model.

This limiting behavior is reminiscent of the transition from a system of particles with two different masses, namely test particles and field particles, which can reach thermodynamic equilibrium since momentum and energy are preserved during collisions, to the Sinai billiard \cite{Sinai1963}, where the scatterers are regarded as infinitely more massive than the moving particles, so that momentum conservation is lost during collisions. The singularity of this limit lies in the fact that a larger mass ratio leads to a smaller exchange of energy and momentum at each collision, which makes equilibration require increasingly long times. In the asymptotic regime of infinite mass ratio, collisions do not preserve momentum, and equilibration does not take place.

Another distinctive feature of the GKR lies in the breaking of color symmetry in the particle–environment interaction, a property called \emph{selectivity}. This mechanism is directly responsible for the emergence of multiple attractors in the phase space of the model, a phenomenon absent in the original Kac ring. Selectivity thus provides an imprint of irreversibility, and in some respects is reminiscent of hysteresis phenomena \cite{Ruelle}. Both the possibility of equilibration between particles and environment, and the appearance of hysteresis-like effects, mark a substantial step beyond the original Kac ring in the modeling of thermodynamic systems.

The paper is organized as follows.
In Sec. \ref{s:model} we introduce the GKR model and outline some general results about the asymptotic behavior of phase space orbits in this model.
Section \ref{s:one} focuses on one-particle systems, whereas Sec. \ref{s:two} analyzes two-particle models.
Future perspectives and conclusions are discussed in Sec. \ref{s:concl}.

\begin{figure}
     \centering
         \includegraphics[width=0.3\textwidth]{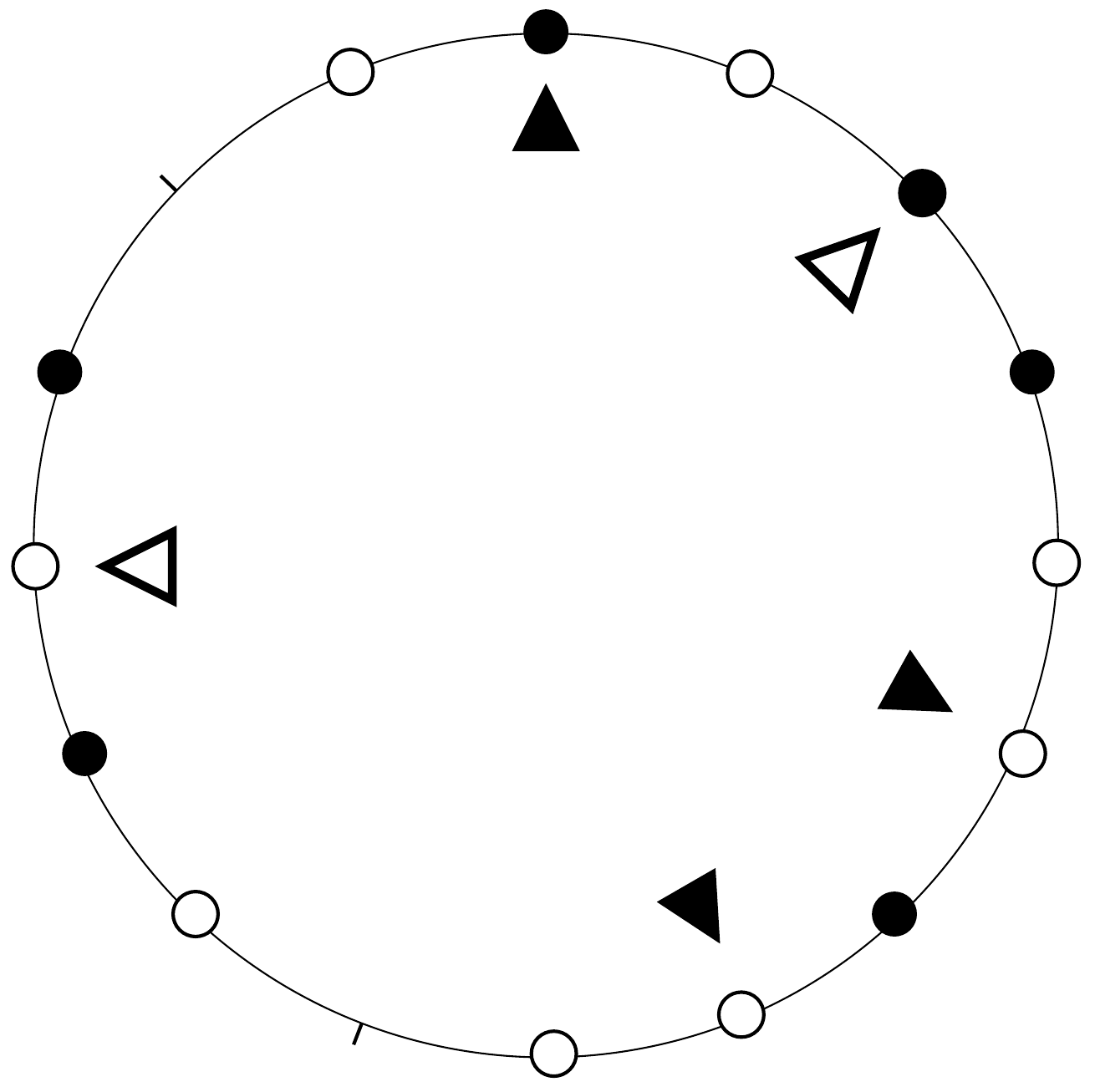}
        \caption{The GKR model consists of a ring with $L$ sites, populated at time $t=0$ by $\hat{N}_b$ black particles and $\hat{N}_w$  white particles (depicted as solid and empty disks in the figure), with $N=\hat{N}_b+\hat{N}_w\le L$. Active and passive scatterers are represented by filled and empty triangles, respectively. At each integer time step, particles move clockwise to the nearest neighboring site on the ring. When a particle encounters an active scatterer, it instantaneously changes color: black becomes white, and white becomes black. In addition, a scatterer also switches its state (from active to passive or vice versa) after undergoing a fixed number of collisions with black particles. This threshold is defined by the parameter $r$, called rigidity.} 
        \label{fig:model}
\end{figure}

\section{Model}
\label{s:model} 
\par\noindent

The GKR model, illustrated in Fig.~\ref{fig:model}, is defined as follows.  
Given positive integers $r$ and $L$, let the \emph{set of states} be  
$Q = \{-1,0,1\} \times \{-1,0,1\} \times \{0,\dots,r-1\}$,  
and let the $L$ point annulus be  
$\Lambda_L = \mathbb{Z}/L\mathbb{Z} = \{0,1,\dots,L-1\}$.  
The corresponding \emph{configuration space} is  
$X_L = Q^{\Lambda_L}$.  

For $x=(x_0,\dots,x_{L-1}) \in X_L$, the element $x_i$ denotes the \emph{state} of site $i$, given by the triple  
$x_i = (o_i, s_i, c_i)$,  
where  
\begin{itemize}
    \item $o_i$ is the \emph{occupation number} of site $i$,  
    \item $s_i$ is the \emph{scatterer state} of site $i$, and  
    \item $c_i$ is the \emph{counter} at site $i$.  
\end{itemize}

The interpretation of these components is as follows:  
\begin{itemize}
    \item $o_i=-1,0,+1$ means that site $i$ is occupied by a white particle, is empty, or is occupied by a black particle, respectively.  
    \item $s_i=-1,0,+1$ indicates, respectively, that site $i$ hosts an active scatterer, no scatterer, or a passive scatterer.  
    \item $c_i \in \{0,\dots,r-1\}$ is the value of the local counter at site $i$.  
\end{itemize}

We consider the discrete time variable $t=0,1,\dots$ and define the deterministic dynamics as follows. At each time step, all particles simultaneously move one site clockwise. Upon arrival, if a particle encounters an active scatterer, its color is instantaneously flipped (from black to white or vice versa). In addition, the local counter at that site is incremented by one (modulo $r$) whenever the incoming particle is black.  

This mechanism embodies the \emph{selectivity} property introduced in Sec.~\ref{s:intro}: only black particles are capable of modifying the environment, by advancing the counters of both active and passive scatterers. White particles, in contrast, leave scatterers unchanged, though they remain affected by active scatterers, as illustrated schematically in Fig.~\ref{fig:inter}.

\begin{figure}
     \centering
         \includegraphics[width=0.45\textwidth]{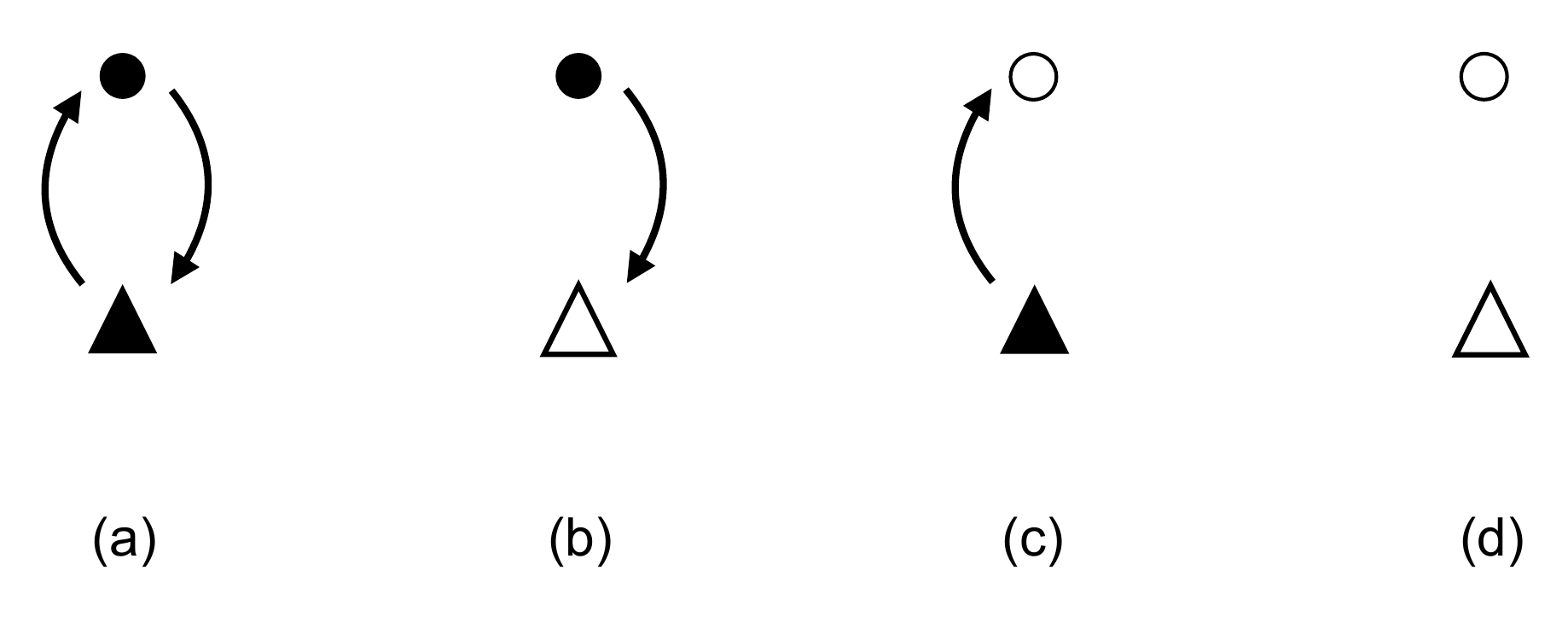}
        \caption{Specific interactions between particles and scatterers in the GKR model: (a) a black particle and an active scatterer, (b) a black particle and a passive scatterer, (c) a white particle and an active scatterer, (d) a white particle and a passive scatterer. An oriented arrow denotes the interaction exerted by the element at the tail on the element at the head, while the absence of an arrow indicates the lack of interaction.} 
        \label{fig:inter}
\end{figure}

The state of a scatterer switches from active to passive, or vice versa, whenever its counter reaches the value $0$.  
The initial condition is chosen such that, at the starting time, all counters are set to zero.  

More algorithmically, the dynamics can be described as follows: for $t\ge1$ and $i\in\Lambda_L$
\begin{itemize}
\item $o_{i}(t)=o_{i-1}(t-1)[1+s_{i}(t-1)-|s_{i}(0)|]$;
\item $c_i(t)=|s_{i}(0)|
              \big[
                c_i(t-1)
                +\delta_{o_{i-1}(t-1),1}
               \big]
                  \!\!\!\mod r$;
\item $s_i(t)=s_i(t-1)
              \big[1
                   -2
                    \delta_{o_{i-1}(t-1),1}
                    \delta_{c_i(t),0}
              \big]$;
\end{itemize}
where $\delta$ is the Kronecker $\delta$ function.
We shall discuss the dynamics with several different initial conditions, 
but we stress that we shall always consider the case 
$c_i(0)=0$ for all $i\in\Lambda_L$.

We note that the dynamics, which, due to the simultaneous updating
rule, constitutes a \emph{Cellular Automaton}, preserves 
the total number of 
particles $N=\sum_{i\in\Lambda_L}|o_i|$ 
and 
the total number of 
scatterers $S=\sum_{i\in\Lambda_L}|s_i|$.
On the other hand, 
the numbers of white and black particles, 
$N_w=\sum_{i\in\Lambda_L}\delta_{-1,o_i}$
and 
$N_b=\sum_{i\in\Lambda_L}\delta_{1,o_i}$, 
as well as 
the numbers of active and passive scatteres, 
$S_a=\sum_{i\in\Lambda_L}\delta_{-1,s_i}$
and 
$S_p=\sum_{i\in\Lambda_L}\delta_{1,s_i}$, 
can change with time.

The time intervals $\{1+kL,2+kL,\dots L+kL\}$, for $k=0,1,\dots$, 
will be called \emph{sweeps}.
 
Our aim is to analyze the evolution and long-term behavior of the GKR. To this purpose, we define a set of observables that describe the macroscopic state of the system.
We begin by defining:
\begin{equation}
    \chi(t) = \frac{N_b(t) - N_w(t)}{N}, \quad \Phi(t) = \frac{S_p(t) - S_a(t)}{S}, \label{orderpam}
\end{equation}
where $\chi(t), \Phi(t) \in [-1, 1]$ for all $t \geq 0$.
Next, for each site $i \in \Lambda_L$, we define:
\begin{equation}
    \sigma_i(t) =
    \begin{cases}
        1 & \text{if } |o_i(t)| = 1 \text{ and } s_i(t) = -1, \\
        0 & \text{otherwise},
    \end{cases}
\end{equation}
so that a nonzero value of $\sigma_i(t)$ signals a color reversal event induced by an active scatterer acting on a particle located at site $i$ at time $t$.
We then define the observable:
\begin{equation}
    \Sigma(t) = \frac{1}{N}\sum_{i \in \Lambda_L} \sigma_i(t), \label{Sigma}
\end{equation}
which measures the number of color reversal events occurring during one complete loop around the ring at time $t$, normalized by the number of particles.
. 



We now state a general result for the GKR model, which includes the original Kac model as a special case. It says that, for any value of the rigidity $r$, all trajectories of the system are eventually periodic, meaning that after a finite transient time any orbit becomes periodic.

\begin{lemma}
For any value of the rigidity $r>0$, all trajectories in the GKR model are eventually periodic.
\label{lemma1}
\end{lemma}

\begin{proof}
Recall that all particles move with one and the same speed (which equals one). Therefore at the moment L, which is equal to the length of the lattice each particle comes to its initial position. Observe that the scatterers do not move and therefore they have fixed positions. The set of all possible states of scatterers is finite. Indeed there is a finite number of scatterers, and each scatterer can be in one of two positions. Therefore there exist some moment of time (which is proportional to the length L of the lattice) when the state of the system repetes itself. Indeed, observe that there is a finite number of states of the particles located in their initial positions (learly, this number equals $2^N$, where $N$ is a number of particles). Also, there is $2^r$ different states of configuration of scatterers, where $r$ is a number of scatterers. Thus, a total number of states of the system which may appear at the moments of time proportional to $L$ is finite. Therefore the statement of lemma holds.
\end{proof}

\begin{remark}
  The limiting periodic state of the model is not unique. It depends on the initial configuration of particles and scatterers and on their initial states. Generally there are many (but a finite number) of the limiting periodic states. The times until a trajectory get into a limiting periodic state can be very large (it generally exponentially depends on the parameters of the model.   
\end{remark}
The behavior of the classical Kac ring is recovered as a special case of Lemma \ref{lemma1}, as follows:
\begin{corollary}
In the original Kac ring model, all trajectories are also eventually periodic.
Indeed, the rigidity of the scatterers in the classical Kac ring is infinite, so their states remain unchanged throughout the evolution.
\end{corollary}

Using Lemma \ref{lemma1}, we can also prove the following result which holds for any observable $O:X_L\rightarrow \mathbb{R}$.
\begin{corollary}
Given $m\in \mathbb{N}$, it holds
\begin{equation}
   \lim_{T \rightarrow \infty} \frac{1}{T} \sum_{t = 1}^{T} O(x(t)) = \frac{1}{m L}\sum_{t=1}^{m L}O(x(t))  \Def\overline{O} \; ,
    \label{O_avg}
\end{equation}
where $m L$ corresponds to the period of a periodic orbit attractor (i.e., the asymptotic periodic orbit).
\end{corollary}
\begin{proof}
It follows from Lemma \eqref{lemma1}, because the total number of all states (configurations) of the considered cellular automaton is finite, that the dynamics becomes periodic at the first moment of time when the system comes to (any) configuration, which already appeared before. After that time, dynamics becomes periodic with some period $m L$, with $m\in \mathbb{N}$, when the system comes back to the configuration which appeared twice.
Hence, all equilibrium states (attractors) of our system are periodic with periods $m L$.
Then the relation \eqref{O_avg} holds because if $T\rightarrow \infty$, then the time interval  $[0,T]$ is constituted by an initial (finite) time till our initial configuration gets to a periodic one (which defines our equilibrium state), a huge number of periods (as $T$ tends to infinity) and, possibly, a noncomplete period corresponding to the last part of the interval $[0,T]$ . Therefore, the averages of the observable $O(x(t))$ over time $T$, when $T$ grows to infinity, is equal to average of the same observable restricted to the time period of the corresponding equilibrium state. 
\end{proof}

From Eqs. \eqref{orderpam} and \eqref{Sigma}, we define the triplet $(\overline{\chi}, \overline{\Phi}, \overline{\Sigma})$, which will serve in the following as a set of macroscopic order parameters for the GKR model. Specifically, while $\overline{\chi}$ and $\overline{\Phi}$  monitor, respectively, the average color of the particles and the average state of the scatterers, $\overline{\Sigma}$ yields the average number of a color reversal events per particle occurring at a given time on the ring $\Lambda_L$ in the equilibrium state.

We observe that periodic orbits consisting of configurations in which all ring
sites are occupied by white particles and all scatterers are passive, i.e.,
$o_i = -1$ and $s_i = 1$ for all $i \in \Lambda_L$, are invariant under the
dynamics and reversible (reversibility follows from inverting the
direction of rotation on the ring). Such special periodic orbits, characterized
by the triplet $(\overline{\chi}, \overline{\Phi}, \overline{\Sigma}) =
(-1,1,0)$, are referred to as \emph{frozen states}. Conversely, periodic orbits
with $\overline{\Sigma} > 0$ are called \emph{oscillating states}.

Finally, we remark that the selectivity property is a key feature of the GKR model responsible for the emergence of frozen states. If black and white particles interacted in the same way with the environment, frozen states would not arise, and the phase-space dynamics would instead settle on periodic orbits corresponding exclusively to oscillating states. In particular, numerical simulations confirm that these periodic orbits include the original configurations of the system, as also observed in the classical Kac model.

\section{Single-Particle Dynamics}
\label{s:one} 
\par\noindent
In this section, we focus on the dynamics of the system introduced in Section~\ref{s:model} in the case where the ring contains only a single particle, $N = 1$. Since in this scenario the sites not occupied by scatterers play no role, we restrict our attention to configurations where $s_i(0) \neq 0$ for all $i \in \Lambda_L$. Recall that in the initial state all counters are set to zero.
Without loss of generality, we assume that the particle is initially located at site $0$, i.e., $o_0(0) \neq 0$ and $o_i(0) = 0$ for all $i \in \Lambda_L \setminus \{0\}$.
This implies that, for any time $t>0$, it holds
$S_a(t)=(1-\Phi(t))L/2$ and $S_p(t)=(1+\Phi(t))L/2$. 

We start our investigation considering the special case in which all scatterers are initially in the same state, i.e. either all active or all passive.
First, we remark that the only non-trivial initial configuration is the one with $o_0(0) = 1$ and $s_i(0) = -1$ for all $i \in \Lambda_L$. Other cases reduce to either trivial or equivalent dynamics:

\begin{itemize}
    \item If $o_0(0) = -1$ and $s_i(0) = 1$ for all $i \in \Lambda_L$, the system is in a \emph{frozen} state from the outset.
    
    \item If $o_0(0) = -1$ and $s_i(0) = -1$ for all $i \in \Lambda_L$, the system becomes equivalent to the non-trivial case after one time step: a black particle appears at site $1$, all scatterers are active, and all counters are reset to zero.
    
    \item If $o_0(0) = 1$ and $s_i(0) = 1$ for all $i \in \Lambda_L$, the system reaches the non-trivial configuration (with $o_0 = 1$ and $s_i = -1$ for all $i$ and all counters zero) after $r$ full sweeps, i.e., after $rL$ time steps.
\end{itemize}

Extensive numerical simulations reveal that the system reaches different periodic orbits depending on the values of the parameters $r$ and $L$. The results are summarized in Table~\ref{tab:table1}, which reports the values of the triplet $(\overline{\chi}, \overline{\Phi}, \overline{\Sigma})$ for $r = 1, \dots, 5$ and $L = 1, \dots, 10$.

\begin{table*}
\begin{tabular}{c|c|c|c|c|c}
       & 1    &  2   & 3    & 4    & 5   \\
\hline
 1  & $(-1.000,1.000,0.000)$ &  $(-1.000,1.000,0.000)$ &  $(-1.000,1.000,0.000)$ &  $(-1.000,1.000,0.000)$ &  $(-1.000,1.000,0.000)$ \\
 2  & $(0.333,-0.333,0.667)$ &  $( 0.143,-0.143,0.571)$ &  $(0.091,-0.091,0.545)$ &  $(0.067,-0.067,0.533)$ &  $(0.053,-0.053,0.526)$ \\
 3  & $(0.143,-0.143,0.571)$ &  $(-1.000,1.000,0.000)$ &  $(-1.000,1.000,0.000)$ &  $(-1.000,1.000,0.000)$ &  $(-1.000,1.000,0.000)$  \\
 4 & $(0.067,-0.067,0.533)$ &  $(-1.000,1.000,0.000)$ &  $(-1.000,1.000,0.000)$ &  $(-1.000,1.000,0.000)$ &  $(-1.000,1.000,0.000)$ \\
 5 & $(-0.048,0.048,0.476)$ &  $(-0.017,0.018,0.491)$ &  $(-1.000,1.000,0.000)$ &  $(-1.000,1.000,0.000)$ &  $(-1.000,1.000,0.000)$  \\
 6 & $(0.016,-0.016,0.508)$ &  $(-1.000,1.000,0.000)$ &  $(-1.000,1.000,0.000)$ &  $(-1.000,1.000,0.000)$ &  $(-1.000,1.000,0.000)$  \\
 7 & $(0.008,-0.008,0.504)$ &  $(0.033,-0.033,0.517)$ &  $(-1.000,1.000,0.000)$ &  $(-1.000,1.000,0.000)$ &  $(-1.000,1.000,0.000)$  \\
 8 & $(-0.175,0.175, 0.413)$ &  $(0.020,-0.020,0.510)$ &  $(-1.000,1.000,0.000)$ &  $(-1.000,1.000,0.000)$ &  $(-1.000,1.000,0.000)$   \\
 9 & $(-0.233,0.228,0.384)$ &  $(-0.004,0.004,0.498)$ &  $(0.003,-0.003,0.501)$ &  $(-1.000,1.000,0.000)$ &  $(-1.000,1.000,0.000)$ \\
 10 & $(-0.001,0.001,0.499)$ &  $(-0.004,0.004,0.498)$ &  $(-1.000,1.000,0.000)$ &  $(-1.000,1.000,0.000)$ &  $(-1.000,1.000,0.000)$ \\
\end{tabular}
\caption{Values of the triplet $(\overline{\chi}, \overline{\Phi}, \overline{\Sigma})$ for $N=1$ and for different values of the rigidity $r$ (horizontal axis) and lattice length $L$ (vertical axis) obtained from numerical simulations of the GKR model. The initial configuration is $o_0(0) = 1$ and $s_i(0) = -1$ for all $i \in \Lambda_L$, corresponding to $\chi(0)=1,\Phi(0)=-1$. Simulations are run over a time interval of $T=10^6$ sweeps, which suffice to reach the attractor for each value of $L,r$ in the table. Numerical values in the table are rounded to the third decimal digit.}
\label{tab:table1}
\end{table*}

\begin{figure}
     \centering
         \includegraphics[width=0.45\textwidth]{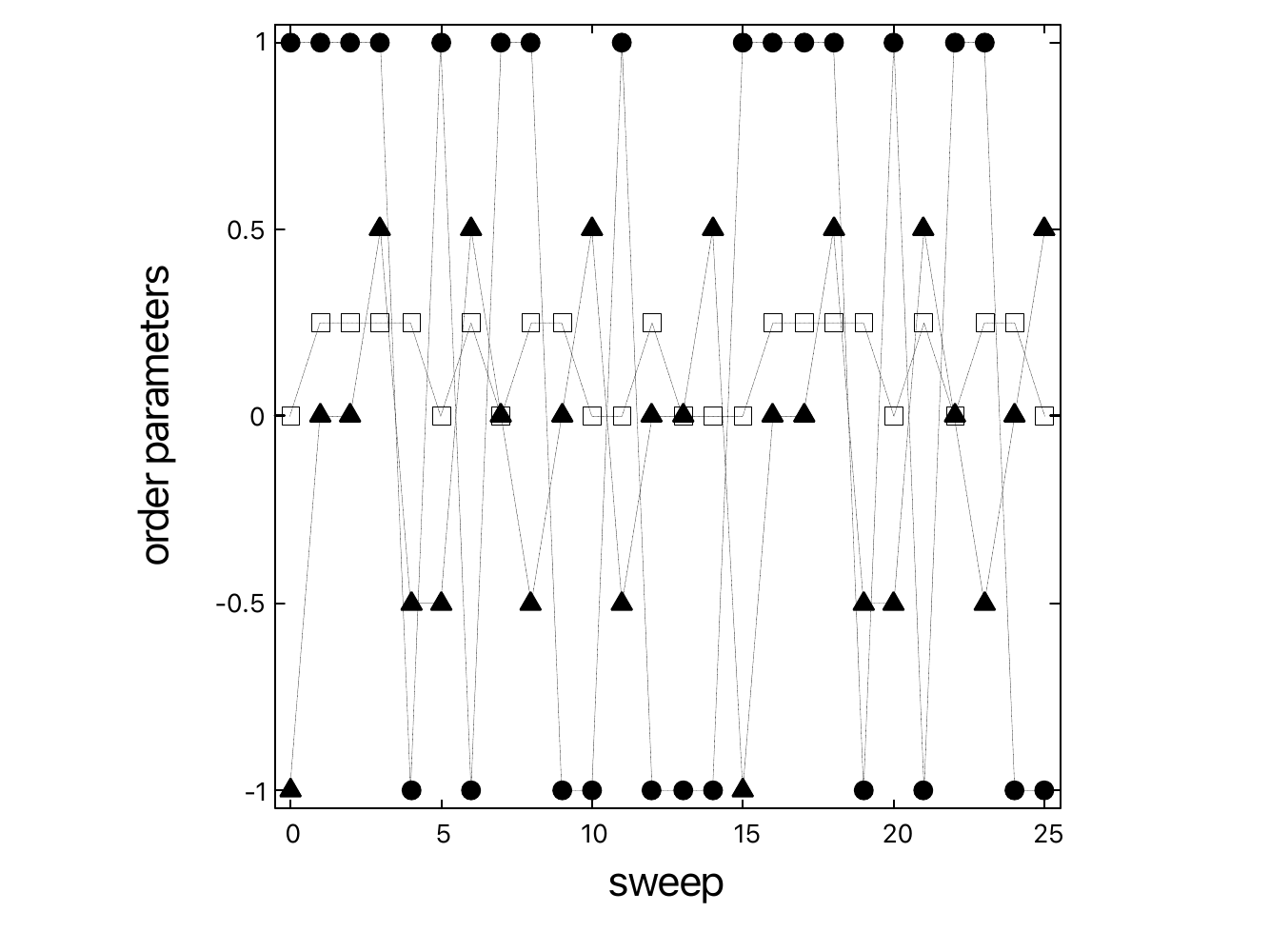}
         \includegraphics[width=0.45\textwidth]{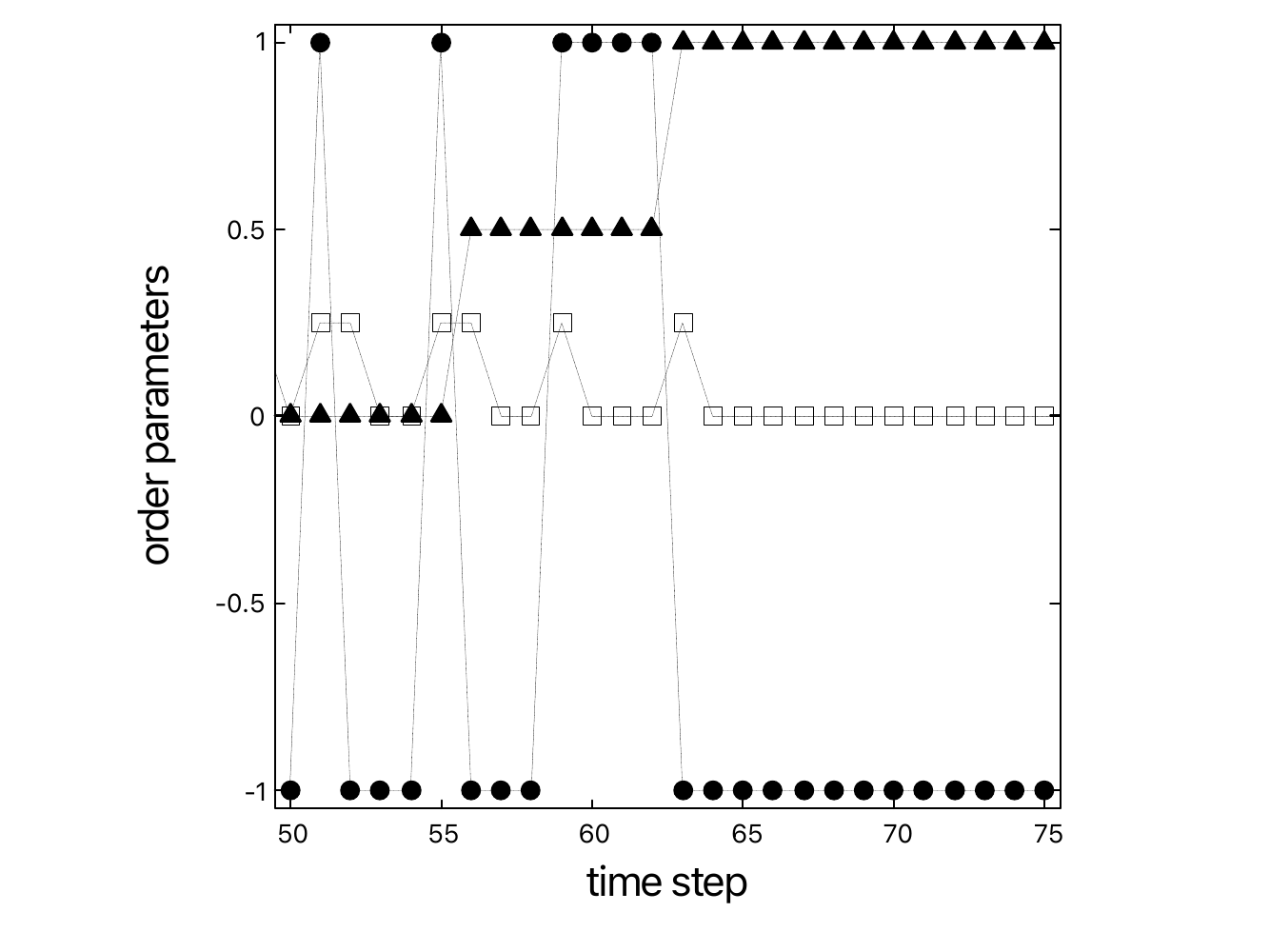}
        \caption{Behavior of $\chi(t)$ (filled disks), $\Phi(t)$ (filled triangles) and $\Sigma(t)$ (empty squares) for the GKR model with $L=4$ and $r=1$ (upper panel) and $r=2$ (lower panel). The initial configuration is the same considered in Table \ref{tab:table1}.} 
        \label{fig:L4}
\end{figure}

The behavior of $\chi(t)$, $\Phi(t)$ and $\Sigma(t)$ for $L=4$ and $r=1,2$ is portrayed in Fig. \ref{fig:L4}. The trajectory corresponding to $r=1$ (upper panel) settles into a periodic orbit of period equal to 15 sweeps, which also includes the initial configuration. Conversely, the trajectory corresponding to $r=2$ (lower panel) ends on a periodic orbit corresponding to a frozen state.
Furthermore, the cases corresponding to $L=1,2$ are amenable to a direct analytical evaluation, which is considered in the following Lemma.

\begin{lemma}
\label{t:one}
Consider the cellular automaton of Section~\ref{s:model} 
with initial state 
$o_0(0)=1$,
$o_i(0)=0$ for all $i\in\Lambda_L\setminus\{0\}$,
$s_i(0)=-1$ 
and
$c_i(0)=0$ for all $i\in\Lambda_L$. 
Then it holds:
\begin{enumerate}
\item
if $L=1$
the frozen state 
$o_0=-1$ 
and 
$s_0=1$ with $c_0=0$ is reached at time $2r-1$. \label{t:one-01}
\item
If $L=2$, an oscillating state with period $8r-2$ is reached at time $2(r-1)$ \label{t:one-02}
\end{enumerate}
\end{lemma}

\par\noindent
\begin{proof}
\ref{t:one-01}.\
Case with $L=1$. If $r=1$, reporting time in the first column and the corresponding 
state of the variables in the remaining ones, we have 
\begin{displaymath}
\begin{array}{c|ccc}
&o_0&c_0&s_0\\
\hline
0&1&0&-1\\
1&-1&0&1\\
\end{array}
\end{displaymath}
which is the invariant state.
If $r\ge2$ we have 
\begin{displaymath}
\begin{array}{c|ccc}
&o_0&c_0&s_0\\
\hline
0&1&0&-1\\
1&-1&1&-1\\
2&1&1&-1\\
\vdots&&&\\
2(r-1)&1&r-1&-1\\
2r-1&-1&0&1\\
\end{array}
\end{displaymath}
which, again, is the invariant state.

\ref{t:one-02}.\
Case with $L=2$.
We start with $r=1$. In this case 
we have the following trajectory: 
\begin{displaymath}
\begin{array}{c|cccccc}
&o_0&c_0&s_0&o_1&c_1&s_1\\
\hline
0&1&0&-1&0&0&-1\\
1&0&0&-1&-1&0&1\\
2&1&0&-1&0&0&1\\
3&0&0&-1&1&0&-1\\
\end{array}
\end{displaymath}
The configuration 
at time $3$ is nothing but 
that at time $0$ with the two sites exchanged. 
Thus, the configuration at time $6$ is equal to that at time $0$.
Hence the trajectory is periodic with period $6$.

We construct, now the trajectory for the case $r=2$. We have 
the following:
\begin{displaymath}
\begin{array}{c|cccccc}
&o_0&c_0&s_0&o_1&c_1&s_1\\
\hline
0&1&0&-1&0&0&-1\\
1&0&0&-1&-1&1&-1\\
2&1&0&-1&0&1&-1\\
3&0&0&-1&-1&0&1\\
4&1&0&-1&0&0&1\\
5&0&0&-1&1&1&1\\
6&-1&1&-1&0&1&1\\
7&0&1&-1&-1&1&1\\
8&1&1&-1&0&1&1\\
9&0&1&-1&1&0&-1\\
\end{array}
\end{displaymath}
The configuration at time $9$ is obtained by that at time by 
exchanging the two sites, thus at time $16$ the configuration 
of time $2$ will be reached. Thus starting from time $2$ the trajectory 
is periodic with period $14$.

Finally, for the case $r\ge3$, in the first part of the trajectory 
the state of the site at $0$ is repeated cyclically. Indeed, we have
\begin{displaymath}
\begin{array}{c|cccccc}
&o_0&c_0&s_0&o_1&c_1&s_1\\
\hline
0&1&0&-1&0&0&-1\\
1&0&0&-1&-1&1&-1\\
2&1&0&-1&0&1&-1\\
3&0&0&-1&-1&2&-1\\
4&1&0&-1&0&2&-1\\
\vdots&&&&&&\\
2(r-1)&1&0&-1&0&r-1&-1\\
\end{array}
\end{displaymath}
In the second part, on the other hand, 
both the sites are changed as follows
\begin{displaymath}
\begin{array}{c|cccccc}
&o_0&c_0&s_0&o_1&c_1&s_1\\
\hline
2r-1&0&0&-1&-1&0&1\\
2r&1&0&-1&0&0&1\\
2r+1&0&0&-1&1&1&1\\
2r+2&-1&1&-1&0&1&1\\
2r+3&0&1&-1&-1&1&1\\
2r+4&1&1&-1&0&1&1\\
\vdots&&&&&&\\
2r+4(r-1)&1&r-1&-1&0&r-1&1\\
6r-3&0&r-1&-1&1&0&-1\\
\end{array}
\end{displaymath}
and at time $6r-3$ the configuration obtained by exchanging the two site 
in the configuration of time $2(r-1)$ is obtained. 
Thus, at time 
$6r-3+[6r-3-2(r-1)]=10r-4$
the same configuration as the one at time $2(r-1)$ is reached. 
Hence, starting from time $2(r-1)$ the trajectory is periodic and the 
period is $10r-4-2(r-1)=8r-2$. 
\end{proof}

To further highlight the presence of multiple attractors in the case $N=1$, we
performed an extensive set of numerical simulations varying the rigidity $r$ and
the ring length $L$. The results are summarized in Table \ref{tab:table1}. In
particular, the first two rows, corresponding to $L=1$ and $L=2$, are
consistent with Lemma~\ref{t:one}, whereas for larger values
$L \in \{1,\dots,10\}$ both oscillating and frozen states are observed as $r$
varies over $\{1,\dots,5\}$. While frozen states tend to dominate at higher
rigidity, the time required to reach the attractor generally increases with $L$
and $r$.

\begin{table*}
\begin{tabular}{c|c|c|c|c|c}
& 1    &  2   & 3    & 4    & 5   \\
\hline
  1  & $(-1.000,1.000,0.000)$ &  $(-1.000,1.000,0.000)$ &  $(-1.000,1.000,0.000)$ &  $(-1.000,1.000,0.000)$ &  $(-1.000,1.000,0.000)$ \\
  2  & $(-1.000,1.000,0.000)$ &  $(-1.000,1.000,0.000)$ &  $(-1.000,1.000,0.000)$ &  $(-1.000,1.000,0.000)$ &  $(-1.000,1.000,0.000)$ \\
 3  & $(-0.040,-0.040,0.480)$ &  $(0.000,0.000,0.500)$ &  $(-0.040,0.040,0.480)$ &  $(-0.059,0.050,0.471)$ &  $(-0.070,0.070,0.466)$  \\
 4 & $(0.067, -0.067,  0.533)$ &  $(-1.000,1.000,0.000)$ &  $(-1.000,1.000,0.000)$ &  $(-1.000,1.000,0.000)$ &  $(-1.000,1.000,0.000)$ \\
 5 & $(0.143, -0.141, 0.571)$ &  $(-0.018, 0.018, 0.491)$ &  $(-0.089, 0.089, 0.456)$ &  $(-0.119, 0.119, 0.440)$ &  $(-0.138, 0.1378, 0.432)$  \\
 6 & $(0.016,-0.016,0.508)$ &  $(-1.000,1.000,0.000)$ &  $(-1.000,1.000,0.000)$ &  $(-1.000,1.000,0.0000$ &  $(-1.000,1.000,0.000)$ \\
  7 & $(0.008, -0.008,  0.504)$ &  $(0.057, -0.057, 0.528)$ &  $(-0.0423 0.043, 0.479)$ &  $(-1.000,  1.000,  0.000)$ &  $(-1.000,  1.000,  0.000)$  \\
  8 & $(0.016, -0.016,  0.508)$ &  $(-0.010,  0.010,  0.495)$ &  $(-1.000,1.000,0.000)$ &  $(-1.000,1.000,0.000)$ &  $(-1.000,1.000,0.000)$ \\
  9 & $(-0.014,  0.014,  0.493)$ &  $(-0.004,  0.004,  0.498)$ &  $(-1.000,1.000,0.000)$ &  $(-1.000,1.000,0.000)$ &  $(-1.000,1.000,0.000)$ \\
  10 & $(-0.001,  0.001,  0.499)$ &  $(-0.004,  0.004,  0.498)$ &  $(-1.000,1.000,0.000)$ &  $(-1.000,1.000,0.000)$ &  $(-1.000,1.000,0.000)$ \\
\end{tabular}
\caption{Values of the triplet $(\overline{\chi}, \overline{\Phi}, \overline{\Sigma})$ for $N=1$ and for different values of the rigidity $r$ (horizontal axis) and lattice length $L$ (vertical axis) obtained from numerical simulations. The initial configuration is $o_0(0) = 1$, $o_i(0) = 0$ for all $i \in \Lambda_L\setminus\{0\}$, and $s_0(0) = 1$, $s_i(0) = -1$ for all $i \in \Lambda_L\setminus \{0\}$. Simulations are run over a time interval of $T=10^6$ sweeps, which suffice to reach the attractor for each value of $L,r$ in the table. Numerical values are rounded to the third decimal digit.}
\label{tab:table2}
\end{table*}

\begin{table*}
\begin{tabular}{c|c|c|c|c|c}
& 1    &  2   & 3    & 4    & 5   \\
\hline
  3  & $(-1.000,1.000,0.000)$ &  $(-1.000,1.000,0.000)$ &  $(-1.000,1.000,0.000)$ &  $(-1.000,1.000,0.000)$ &  $(-1.000,1.000,0.000)$ \\
   4 & $(0.067, -0.067,  0.533)$ & $(-1.000,1.000,0.000)$ &  $(-1.000,1.000,0.000)$ &  $(-1.000,1.000,0.000)$ &  $(-1.000,1.000,0.000)$ \\
 5 & $(-0.048, 0.048, 0.476)$ &  $(-0.130, 0.129, 0.435)$ &  $(-0.155,  0.153,  0.422)$ &  $(-1.000,  1.000,  0.000)$ &  $(-0.174,  0.172,  0.413)$  \\
 6 & $(0.016, -0.016,  0.508)$ &  $(-1.000,1.000,0.000)$ &  $(-1.000,1.000,0.000)$ &  $(-1.000,1.000,0.000)$ &  $(-1.000,1.000,0.000)$ \\
  7 & $(0.008, -0.008,  0.504)$ &  $(-0.005,  0.005,  0.498)$ &  $(0.000,  0.000,  0.500)$ &  $(-1.000,1.000,0.000)$ &  $(-1.000,1.000,0.000)$ \\
  8 & $(0.079, -0.079, 0.540)$ &  $(-0.010,  0.010,  0.495)$ &  $(-1.000,1.000,0.000)$ &  $(-1.000,1.000,0.000)$ &  $(-1.000,1.000,0.000)$ \\
  9 & $(-0.014,  0.014,  0.493)$ &  $(-0.004,  0.004,  0.498)$ &  $(-1.000,1.000,0.000)$ &  $(-1.000,1.000,0.000)$ &  $(-1.000,1.000,0.000)$ \\
  10 & $(0.008, -0.008, 0.504)$ &  $(-0.00, 0.0044, 0.498)$ &  $(-1.000,1.000,0.000)$ &  $(-1.000,1.000,0.000$ &  $(-1.000,1.000,0.000)$ \\
\end{tabular}
\caption{Values of the triplet $(\overline{\chi}, \overline{\Phi}, \overline{\Sigma})$ for $N=1$ and for different values of the rigidity $r$ (horizontal axis) and lattice length $L$ (vertical axis) obtained from numerical simulations. The initial configuration is $o_0(0) = 1$, $o_i(0) = 0$ for all $i \in \Lambda_L\setminus\{0\}$, and $s_0(0) = s_2(0)= 1$, $s_i(0) = -1$ for all $i \in \Lambda_L\setminus \{0,2\}$. Simulations are run over a time interval of $T=10^6$ sweeps, which suffice to reach the attractor for each value of $L,r$ in the table. Numerical values are rounded to the third decimal digit.}
\label{tab:table3}
\end{table*}

We also tested the sensitivty of the GKR model to the initial configuration of the
scatterers. For instance, Table \ref{tab:table2} reports the triplet
$(\overline{\chi}, \overline{\Phi}, \overline{\Sigma})$ for the case
$s_0(0) = 1$ and $s_i(0) = -1$ for $i \in \Lambda_L \setminus \{0\}$. In particular,
when $L = 2$ and $r \in \{1, \dots, 5\}$, the system exhibits frozen states,
in contrast to the results of Lemma \ref{t:one} shown in the second row of
Table \ref{tab:table1}, revealing that only oscillating states instead emerge when
$s_i(0) = -1$ for all $i \in \Lambda_L$.
Analogously, Table~\ref{tab:table3} reports the values of the triplet for the
case $s_0(0) = s_2(0) = 1$ and $s_i(0) = -1$ for
$i \in \Lambda_L \setminus \{0,2\}$. In this setting, a mixture of oscillating
and frozen states is observed for all considered values of
$L \in \{1,\dots,5\}$ and $r \in \{1,\dots,5\}$.

Furthermore, we altered the direction of particle motion (from clockwise to
anticlockwise) to probe the time reversibility of the dynamics on an attractor.
For frozen states, where the rigidity of the scatterers is irrelevant, reversing
the rotation direction allows the system to retrace exactly all previous
configurations of the periodic orbit. This is not the case for oscillating
states, where reversing the rotation typically drives the dynamics toward a
different periodic orbit than the one reached for the same pair $(L,r)$ under
clockwise dynamics (see Table \ref{tab:table1}).

To gain further insight, we performed the numerical experiment illustrated in
Fig. \ref{fig:rev}, considering the GKR model with $L=5$, $r=1$. We started from
the initial configuration $o_0(0) = 1$ and $s_i(0) = -1$ for all
$i \in \Lambda_L$, which already belongs to a periodic orbit corresponding to
an oscillating state with period $21$ sweeps. We then ran the automaton until
the starting configuration was reached again at time $t_* = 105$, thereby
recording an entire cycle of the periodic orbit
$\{x(t)\}_{t=0}^{t_*}$. From this cycle we reconstructed the corresponding
time-reversed dynamics, given by the backward sequence $\{x(t_* - t)\}_{t=0}^{t_*}$.

Next, we used the same initial configuration as the starting state for
anticlockwise dynamics on the ring. In this case, the initial configuration also
belongs to a periodic orbit of the anticlockwise dynamics, again an oscillating
state with the same period but distinct from the one reached under clockwise
dynamics. The comparison between the time-reversed and anticlockwise dynamics is
shown in Fig.~\ref{fig:rev}, and reveals that the behavior of $\chi(t)$ differs
between the two evolutions. However, in the case considered in
Fig.~\ref{fig:rev}, the triplet $(\overline{\chi}, \overline{\Phi},
\overline{\Sigma})$ coincided (up to the third decimal digit) for clockwise and
anticlockwise dynamics. Table~\ref{tab:table4} summarizes the results of this
numerical experiment for different pairs $(L,r)$. Boxed entries highlight the
cases where the triplet $(\overline{\chi}, \overline{\Phi}, \overline{\Sigma})$
changes when switching from clockwise dynamics (see Table \ref{tab:table1}) to
anticlockwise particle dynamics on the ring.

\begin{figure}
     \centering
         \includegraphics[width=0.45\textwidth]{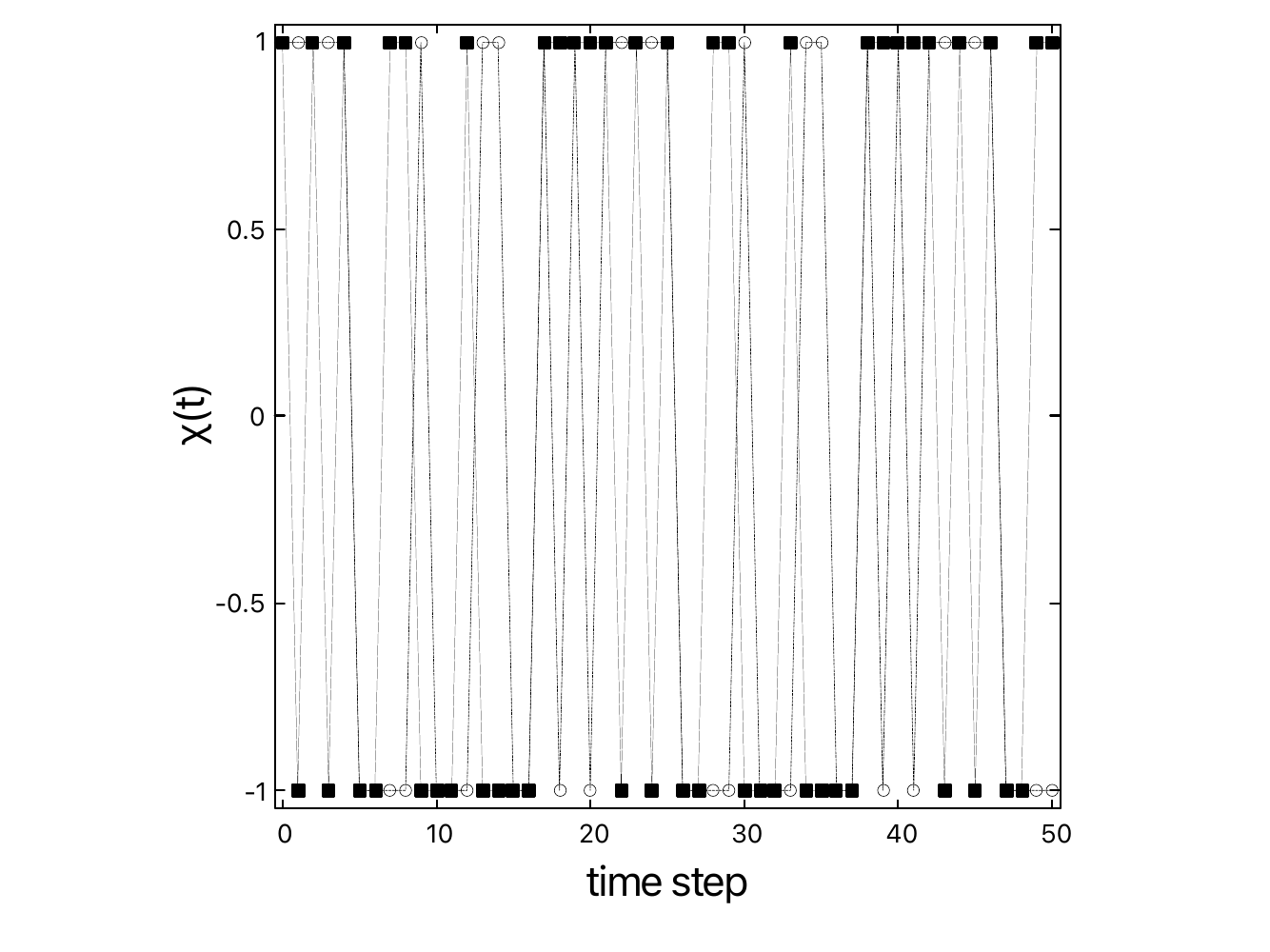}
        \caption{Behavior of $\chi(t)$ for the GKR model with $L=5$ and $r=1$ for time-reversed (empty disks) and anticlockwise dynamics (filled squares). In both evolutions the initial configuration, which belongs to a periodic orbit, is $o_0(0) = 1$, $s_i(0) = -1$ for all $i \in \Lambda_L$.} 
        \label{fig:rev}
\end{figure}


\begin{table*}
%
\begin{tabular}{c|c|c|c|c|c}
 & 1    &  2   & 3    & 4    & 5   \\
 \hline

  1  & $(-1.000,1.000,0.000)$ &  $(-1.000,1.000,0.000)$ &  $(-1.000,1.000,0.000)$ &  $(-1.000,1.000,0.000)$ &  $(-1.000,1.000,0.000$ \\
  2  & $(0.333,-0.333,0.667)$ &  $( 0.143,-0.143,0.571)$ &  $(0.091,-0.091,0.545)$ &  $(0.067,-0.067,0.533)$ &  $(0.053,-0.053,0.526)$ \\
 3  & $(0.143,-0.143,0.571)$ &  $(-1.000,1.000,0.000)$ &  $(-1.000,1.000,0.000)$ &  $(-1.000,1.000,0.000)$ &  $(-1.000,1.000,0.000)$  \\
   4 & $(0.067,-0.067,0.533)$ &  $(-1.000,1.000,0.000)$ &  $(-1.000,1.000,0.000)$ &  $(-1.000,1.000,0.000)$ &  $(-1.000,1.000,0.000)$ \\
 5 & $(-0.048,0.048,0.476)$ &  $(-0.018,0.018,0.491)$ &  $(-1.000,1.000,0.000)$ &  $(-1.000,1.000,0.000)$ &  $(-1.000,1.000,0.000)$  \\
 6 & $(0.016,-0.016,0.508)$ &  $(-1.000,1.000,0.000)$ &  $(-1.000,1.000,0.000)$ &  $(-1.000,1.000,0.000)$ &  $(-1.000,1.000,0.000)$  \\
  7 & $(0.008,-0.008,0.504)$ &  $\fbox{(-0.129,0.129,0.436)}$ &  $(-1.000,1.000,0.000)$ &  $(-1.000,1.000,0.000)$ &  $(-1.000,1.000,0.000)$  \\
  8 & $(-0.175,0.175, 0.413)$ &  $\fbox{(-0.010,0.010,0.495)}$ &  $(-1.000,1.000,0.000)$ &  $(-1.000,1.000,0.000)$ &  $(-1.000,1.000,0.000)$   \\
  9 & $(-0.233,0.228,0.384)$ &  $\fbox{(0.096,-0.097,0.548)}$ &  $\fbox{(0.016,0.097,0.451)}$ &  $(-1.000,1.000,0.000)$ &  $(-1.000,1.000,0.000)$ \\
  10 & $(-0.001,0.001,0.49)$ &  $(-0.0044,0.004,0.498)$ &  $(-1.000,1.000,0.000)$ &  $(-1.000,1.000,0.000)$ &  $(-1.000,1.000,0.000)$ \\
\end{tabular}
\caption{Values of the triplet $(\overline{\chi}, \overline{\Phi}, \overline{\Sigma})$
for $N=1$ and different values of the rigidity $r$ (horizontal axis) and lattice length $L$ (vertical axis),
obtained from numerical simulations of the GKR model with anticlockwise particle
dynamics on the ring. For each pair $(L,r)$, the initial configuration was taken
from the corresponding periodic orbit reached after $10^6$ sweeps with clockwise
dynamics. Numerical values are rounded to the third decimal digit. Boxed
entries highlight the cases in which the triplet differs from the values
obtained under clockwise evolution, reported in Table~\ref{tab:table1}.
}
\label{tab:table4}
\end{table*}

\begin{table*}
\begin{tabular}{c|c|c|c|c|c}
 & 1    &  2   & 3    & 4    & 5   \\
 \hline
   2 & $(-1.000,  1.000,  0.000)$ &  $(-1.000,  1.000,  0.000)$ &  $(-1.000,  1.000,  0.000)$ &  $(-1.000,  1.000,  0.000)$ &  $(-1.000,  1.000,  0.000)$ \\
   3 & $(0.143, -0.143,  0.571)$ &  $(0.143, -0.143,  0.571)$ &  $(0.059, -0.098,  0.529)$ &  $(-1.000,  1.000,  0.000)$ &  $(0.035, -0.057,  0.517)$ \\
   4 & $(0.000, -0.167,  0.500)$ &  $(0.067, -0.067,  0.533)$ &   $(-1.000,  1.000,  0.000)$ &  $(-1.000,  1.000,  0.000)$ &  $(-1.000,  1.000,  0.000)$ \\
 5 & $(0.032, -0.032,  0.516)$ &  $(0.032, -0.032,  0.516)$ &  $(-1.000,  1.000,  0.000)$ &  $(-0.018,  0.018,  0.491)$ &  $(-1.000,  1.000,  0.000)$ \\
 6 & $(-0.143, -0.048  0.429)$ &  $(-1.000,  1.000,  0.000)$ &  $(-0.034,  0.034  0.483)$ &  $(-1.000,  1.000,  0.000)$ &  $(-1.000,  1.000,  0.000)$ \\
  7 & $(-0.011,  0.011,  0.495)$ &  $( 0.040, -0.040,  0.520)$ &  $(0.003, -0.003,  0.501)$ &  $(-1.000,  1.000,  0.000)$ &  $(-1.000,  1.000,  0.000)$  \\
  8 & $(-0.200,  0.000,  0.400)$ &  $(0.020, -0.020,  0.255)$ &  $(-1.000,  1.000,  0.000)$ &  $(-1.000,  1.000,  0.000)$ &  $(-1.000,  1.000,  0.000)$ \\
  9 & $(-0.002,  0.002,  0.499)$ &  $(0.01, -0.0091,  0.508)$ &  $(-1.000,  1.000,  0.000)$ &  $(-1.000,  1.000,  0.000)$ &  $(-1.000,  1.000,  0.000)$ \\
  10 & $(-0.286,  0.094,  0.357)$ &  $(-0.018,  0.018,  0.491)$ &  $(-1.000,  1.000,  0.000)$ &  $(-1.000,  1.000,  0.000)$ &  $(-1.000,  1.000,  0.000)$ \\
\end{tabular}
\caption{Values of the triplet $(\overline{\chi}, \overline{\Phi}, \overline{\Sigma})$ for $N=2$ and for different values of the rigidity $r$ (horizontal axis) and lattice length $L$ (vertical axis) obtained from numerical simulations. The initial configuration is $o_i(0) = o_1(0) = 1$, $o_i(0)=0$ for all $i \in \Lambda_L\setminus\{0,1\}$, $s_i(0) = -1$ for all $i \in \Lambda_L$. Simulations are run over a time interval of $T=10^6$ sweeps, which suffice to reach the attractor for each value of $L,r$ in the table. Numerical values are rounded to the third decimal digit. Particles move clockwise on the ring.}
\label{tab:table5}
\end{table*}

\section{Two-particle case}
\label{s:two} 
\par\noindent

In this section we fix $N=2$ to probe the effect of particle interactions mediated by the scatterers.
The results of our numerical analysis for $N=2$, $L \in {2,\dots,5}$, and $r \in {1,\dots,5}$ are reported in Tab.~\ref{tab:table5}, showing the emergence of new periodic orbits (see Tab.~\ref{tab:table1} for comparison).
In particular, unlike the case $L=2$ and $N=1$ discussed in case~\ref{t:one-02} of Lemma~\ref{t:one}, for $L=2$ and $N=2$ particle interactions produce periodic orbits that consist entirely of frozen states, as established in the following Lemma.
\begin{lemma}
\label{t:two}
Consider the cellular automaton of Section~\ref{s:model} with $L=2$, $N=2$, and
initial state 
$o_0(0)=o_1(0)=1$, 
$s_0(0)=s_1(0)=-1$, 
and $c_0(0)=c_1(0)=0$. 
Then, for any $r \ge 1$, a frozen state is reached at time $2r-1$.
\end{lemma}

\begin{proof}
Since the particles are indistinguishable and the scatterers are identical, it suffices to focus on a single scatterer. 
This case can therefore be reduced to case~\ref{t:one-01} of Lemma~\ref{t:one}, 
corresponding to the GKR model with $L=1$.
\end{proof}
An extension to specific GRK models with $L\ge 1$ and $N=L$ comes with the following Corollary.

\begin{corollary}
The same reasoning as the one adopted in the proof of Lemma~\ref{t:two} applies to any GKR model with $L \ge 1$, $N=L$ with initial configuration $o_i(0)=1$, $c_i(0)=0$ and $s_i(0)=-1$ for all $i \in \Lambda_L$.
\end{corollary}

\section{Conclusion}
\label{s:concl}

The classical Kac ring is widely recognized as a highly idealized yet paradigmatic model for studying the emergence of irreversibility from microscopic reversible dynamics. In this paper, we introduce and analyze a more general variant, the Generalized Kac Ring (GKR) model. Unlike the original Kac ring, where the environment consists of inert scatterers, the GKR incorporates scatterers with internal states that evolve through interactions with the particles. This makes the GKR self-consistent, coupling particle and environment dynamics within a unified deterministic framework. In particular, the model includes a feedback mechanism in which particles modify the scatterers, which in turn influence particle motion, generating indirect and non-local interactions without direct collisions (cf.~\cite{Colangeli2023} for a continuous-time billiard variant). Such a mechanism aligns the GKR with several classes of physical systems in which mobile carriers and the medium co-evolve, potentially on different time scales.
This scenario is commonly found, for instance, in solid-state physics of electrical conductors or in plasma physics, where electrons are often treated as interacting only with the ions which constitute their medium. Photons in a black-body cavity provide another analogy, as they interact with the atoms of the walls but not directly with each other.
In many physical contexts, these interactions do not alter the state of the scatterers, which are much heavier than the particles. However, the effect depends on the microscopic details. For example, a single electron–ion collision may have negligible impact on the ions, but repeated collisions can alter their state and eventually lead to thermal equilibration of the two species, passing through a transient stage in which electrons influence ions and the evolving ion states elicit a feedback on the electrons. The smaller the electron-to-ion mass ratio, the longer this equilibration requires. In our framework, this feature is codified by the rigidity of the scatterers. Quantum effects may introduce additional channels of indirect interaction, such as excitation or ionization of the heavier particles.
Other analogies arise in various toy models of statistical mechanics. Examples include kinetic Ising-type models and spintronics, where conduction electrons exchange angular momentum with localized spins, thereby modifying the medium’s magnetization and altering subsequent transport properties \cite{Kubo1966,Ohno2014}. Models of annealed disorder \cite{Bouchaud1990} provide another instance, where transport occurs through a medium that adapts dynamically to the carriers. A classical example comes from polaron physics, where charge carriers distort the surrounding lattice, and the modified lattice in turn alters the carrier dynamics \cite{Alexandrov2010}. Beyond interacting particle systems, reaction–diffusion models exhibit similar features, as mobile reactants can alter the state of catalytic sites, thereby influencing future reaction pathways \cite{vanKampen1992,Dickman1999}. A simple physical example is a mixture of two gases of comparable molecular mass, where one gas is sufficiently rarefied that self-collisions are negligible.
It should be emphasized that not all interactions are equivalent. Some permit the establishment of local thermodynamic equilibrium, while others do not, depending sensitively on the environment in which they occur. In highly confining media, for instance, anomalous transport of matter or energy is often observed, as in single-file diffusion or in the Knudsen regime. Furthermore, indirect interactions imposed by Gaussian or Nos\`{e}-Hoover deterministic thermostats do not guarantee convergence to thermodynamic behavior. The reliable foundation of thermodynamics remains in systems of particles with short-range repulsive cores, possibly complemented by short-range attractive tails.
These caveats are not of concern here, as our goal is to extend the classical Kac ring model to a setting where interactions play a more central role. We have shown that the phenomenology of our model is richer and that certain parameters, such as the rigidity of the scatterers, which control the two time scales of particle and scatterer evolution, determine the temporal evolution of the system. Remarkably, the resulting behaviors reveal the emergence of multiple attractors, corresponding to periodic orbits, thereby reflecting, in a highly idealized setting, a diversity observed in systems of physically relevant interacting constituents.

\begin{acknowledgments}
ENMC thanks the PRIN 2022 project ``Mathematical modelling of heterogeneous systems'' (code: 2022MKB7MM, CUP: B53D23009360006). 
The research of MC has been developed in the framework of the Research Project
National Centre for HPC, Big Data and Quantum Computing-PNRR Project, funded by the
European Union-Next Generation EU. 
LR gratefully acknowledges support from the Italian Ministry of University and Research (MUR) through the grant
PRIN2022-PNRR project (No. P2022Z7ZAJ) ``A Unitary Mathematical Framework for Modelling Muscular Dystrophies''
(CUP: E53D23018070001). This research was performed under the auspices of Italian National Group of Mathematical
Physics (GNFM) of the National Institute for Advanced Mathematics - INdAM.
\end{acknowledgments}


\bibliographystyle{unsrt}

\bibliography{bccr-kac}

\end{document}